\documentclass[submission,copyright,creativecommons,UKenglish]{eptcs}
\providecommand{\event}{EXPRESS/SOS 2024}

\usepackage{xspace}
\usepackage{xcolor}
\usepackage{graphicx}
\usepackage{stmaryrd, amsmath, amssymb, amsthm}
\usepackage{bussproofs}
\usepackage{paralist}
\usepackage{shuffle}
\usepackage{multirow}
\usepackage{tikz}
\usetikzlibrary{shapes}
\usetikzlibrary{automata}
\usetikzlibrary{calc}
\usepackage{subcaption}

\usepackage{macro}
\usepackage[capitalise]{cleveref}
\usepackage{cancel}

\title{Synchronisability in Mailbox Communication}

\author{Cinzia {Di Giusto}
\institute{Université Côte d'Azur,\\ CNRS,I3S, France}
\and
Laetitia Laversa 
\institute{Université Sorbonne Paris Nord,\\ Paris, France}
\and
Kirstin Peters
\institute{Universität Augsburg,\\ Augsburg, Germany}
}

\def\titlerunning{Synchronisability in Mailbox Communication}
\def\authorrunning{C. {Di Giusto}, L. Laversa \& K. Peters}

\begin{document}

\maketitle


\begin{abstract}
We revisit the problem of synchronisability for communicating automata, \ie whether the language of send messages for an asynchronous system is the same as the language of send messages with a synchronous communication. The un/decidability of the problem depends on the specific asynchronous semantics considered as well as the topology (the communication flow) of the  system. 
Synchronisability is known to be undecidable under the peer-to-peer semantics, while it is still an open problem for mailbox communication. The problem was shown to be decidable for ring topologies. In this paper, we show that when generalising to automata with accepting states, synchronisability is undecidable under the mailbox semantics, this result is obtained by resorting to the Post Correspondence problem. In an attempt to solve the specific problem where all states are accepting, we also show that synchronisability is decidable for tree topologies (where, as well as for rings, peer-to-peer coincides with mailbox semantics).
We also discuss synchronisability for multitrees in the mailbox setting.
\end{abstract}


\section{Introduction}


Communicating automata~\cite{brand_communicating_1983}, \ie a network of finite state automata (the participants) where transitions can be interpreted as sending or receiving actions, are a common way to model communication protocols.
This model allows to consider systems exchanging messages in either a synchronous or asynchronous way. While in the former case communication happens simultaneously, in the latter, messages are sent to buffers where they wait until they are received by other participants. Several semantics have been proposed in the literature, \eg  \cite{charron_synchronous_1996,chevrou_diversity_2016,digiusto_partial_2023}.
Nonetheless, the two most prominent ones are peer-to-peer (\Ptp) communication (where between each pair of participants there are two FIFO buffers, one per direction of communication) and  \Mailbox communication (where each participant has its own FIFO buffer that stores all received messages, whatever the sender). 
\Ptp is more generally found in channel-based languages (\eg Go, Rust), while \Mailbox is more common on actor languages such as Erlang or Elixir.

From the expressiveness point of view, when considering asynchronous communicating automata, Turing machines can be encoded with  two participants and two FIFO buffers only \cite{brand_communicating_1983}.  On the other side of the spectrum, synchronous systems are equivalent to finite state automata.
In order to fill the gap between these two extremes and recover some decidability,  several approaches have been considered to approximate the synchronous behaviour. Results can be classified into two main families of systems. In the first one,  asynchronous behaviours are limited by bounding the size of   buffers. While in the second, asynchronous behaviours are bounded by only considering systems where all the executions can be related  (up to different equivalence relations) to synchronous executions. 
%
Some examples of the first family  are existentially ($\exists$) and universally ($\forall$) $B$-bounded systems \cite{genest_kleene_2006}. A system is universally $B$-bounded if all its executions are $B$-bounded, \ie can be made with buffers of size $B$, and existentially $B$-bounded if all its executions are causally equivalent to a $B$-bounded one. 
If a system is $\exists$/$\forall$-$B$-bounded, model checking problems (\ie checking whether a configuration is reachable or more generally whether a monadic second order formula is satisfied) turn out to be decidable. Unfortunately deciding whether a given system is $\exists$/$\forall$-$B$-bounded, when $B$ is unknown is undecidable for \Ptp \cite{kuske_communicating_2010} and  \Mailbox \cite{bollig_unifying_2021}.
Instead, as an example of  the second family, in \cite{bouajjani_completeness_2018}, the authors define the class of $k$-synchronisable systems, which requires that any execution is causally equivalent to an execution that can be divided into slices of $k$ messages, where all the sending must be done before the receiving. In these systems, the reachability of a configuration is decidable. Moreover, deciding whether a given \Ptp system is $k$-synchronisable is decidable \cite{bouajjani_completeness_2018}, and the same for a given \Mailbox system \cite{digiusto_k-synchronizability_2020}. It is also decidable whether there exists a $k$ such that a given \Mailbox system is $k$-synchronisable \cite{giusto_guessing_2021}. 

The causal equivalence relation \cite{digiusto_k-synchronizability_2020} derives from the \emph{happened before} relation \cite{lamport_time_1978}, which ensures that actions are performed in the same order, from the point of view of each participant. It allows to compare and group executions into classes. In \cite{basu_deciding_2016}, the authors define a different notion of synchronisability, which does not rely on causal equivalence, but on send traces (the projection of executions onto send actions). A system is synchronisable if each asynchronous execution has the same send trace as a synchronous execution of the system. 
This differs from k-synchronisability, since actions may be performed in a different order by a participant, and the class of k-synchronisable systems and the one of synchronisable systems are incomparable. 
Synchronisability of a system implies that reachability is decidable in it. 
A way of checking if a system is synchronisable (for \Ptp and \Mailbox) was proposed
in \cite{basu_deciding_2016}: the authors claimed that if the set of synchronous send traces is equal to the set of 1-bounded send traces, then the system is synchronisable. The claim is actually false, as shown in \cite{finkel_synchronizability_2017}. The authors  provide two counterexamples showing that the method is faulty for both \Ptp and \Mailbox.  The counterexamples in \cite{finkel_synchronizability_2017} expose cases where the set of 2-bounded send traces contains traces that do not exist in the set of 1-bounded send traces, while the latter are identical to the set of synchronous send traces.
Moreover, checking synchronisability of a system communicating with a  \Ptp architecture, is shown to be undecidable in \cite[Theorem 3]{finkel_synchronizability_2017}, while the  problem remains open for \Mailbox systems. 

\paragraph{Contributions.} In this paper we start answering this last question. In a first attempt to assess the problem, we relax one of the hypothesis in \cite{finkel_synchronizability_2017} and consider the general case of communicating automata with final accepting states. This allows us to code the Post Correspondence Problem into our formalism and thus prove undecidability of what we call the  Generalised Synchronisability Problem.
Final states are a key ingredient in the proof which cannot be  adapted to the case without accepting states (or said otherwise where all states are accepting).

To understand where the expressiveness of the problem lies, we started considering how different topologies of communication --\ie the underlining structure of exchanges-- affects the decidability of the Synchronisability Problem.
Our first step generalises another result of \cite{finkel_synchronizability_2017}, where it was shown that Synchronisability Problem is decidable for oriented ring topologies, \cite[Theorem 11]{finkel_synchronizability_2017}. Here, we consider trees and show that the Synchronisability Problem is decidable. The result is obtained by showing that the language of buffers is regular and can be computed. We believe that our algorithm can be extended to multitrees (which are acyclic graphs where among each pair of nodes there is a single path). 
Notice that our approach is more direct than the one in \cite{basu_deciding_2016}.
Instead of comparing (regular) sets of send traces we directly analyse the content of buffers
 
\paragraph{Outline.}
The paper is organised as follows: \cref{sec:prelim} introduces the necessary terminology. The first undecidability result is given in  \cref{sec:post}, while \cref{sec:decidability} discuss the decidability of synchronisability for tree topologies. 
Finally, \cref{sec:concl} concludes with some perspectives.


\section{Preliminaries}
\label{sec:prelim}

For a finite set $\alphabet$, a word $w= a_1a_2\dots a_n \in \alphabet^*$  is a finite sequence of symbols, such that $a_i \in \alphabet$,  for all $1 \leq i \leq n$. 
The concatenation of two words $\word_1$ and $\word_2$ is denoted $\word_1\concat \word_2$, $\length{\word}$ denotes the length of $\word$,
and $\emptyword$ denotes the empty word. We assume some familiarity with non-deterministic finite state automata,
and we write $\languageof{\Automaton{}}$ for the language accepted
by  automaton $\Automaton{}$. 

A communicating automaton is a finite state machine that performs actions of two kinds: either sends or receives. 
A \emph{network} of communicating automata, or simply network, is the parallel composition of a finite set $\Part$ of participants that  exchange messages.  We consider a finite set of messages \messageSet. Each message in \messageSet consists of a sender, a receiver, and some finite information.
We denote $\msg{a}{p}{q} \in \messageSet$ the message sent from peer $p$ to  $q$ with payload $a$, $p \neq q$, \ie a peer can not send/receive messages to/from itself.
An action is the send $!m$  or the reception $?m$ of a message $m \in \messageSet$. We denote the set of actions for peer $p$ $\actionSet{p} = \{ \send{a}{p}{q}, \rec{a}{q}{p} \mid \msg{a}{p}{q}\in \messageSet \wedge q \in \Part \}$  and $\sendact{p}  = \{ \send{a}{p}{q} \mid \msg{a}{p}{q}\in \messageSet \wedge ~q \in \Part \}$ the set of sends from $p$.

\begin{definition}[Network of communicating automata]
	\label{def:network}
	$\Network = ( (\Automaton{p})_{p \in \Part}, \messageSet)$ is a network of communicating automata, where:
	\begin{enumerate}
		\item for each $p \in \Part$, $\Automaton{p}=(\stateSet{p}, \state{s}{0}{p}, \messageSet,  \autTransition{p}, \finalStates{p})$ is a communicating automaton with $\stateSet{p}$ is a finite set of states,  $\state{s}{0}{p} \in \stateSet{p}$ the initial state, $\autTransition{p} \ \subseteq \stateSet{p} \times \actionSet{p} \times \stateSet{p}$ is a transition relation, and $\finalStates{p}$ a set of final states and
		\item for each $ m \in \messageSet $, there are $ p \in \Part $ and $ s_1, s_2 \in \stateSet{p} $ such that $ \left(s_1, !m, s_2 \right) \in {\autTransition{p}} $ or $ \left(s_1, ?m, s_2 \right) \in {\autTransition{p}} $.
	\end{enumerate}
\end{definition}

The topology of a system is a graph with arrows from senders to receivers.

\begin{definition}[Topology]
	\label{def:topology}
	Let $\Network=((\Automaton{p})_{p\in \Part}, \messageSet)$ be a network of communicating automata. Its topology is an oriented graph $\topo{\Network} = (\vertexSet,\edgeSet) $, where $\vertexSet=\{p \mid p \in  \Part\}$ and $\edgeSet = \left\{ (p, q) \mid \exists \msg{a}{p}{q} \in \messageSet \right\}$.
\end{definition}

Let $\senders{p}$ (\resp $\receivers{p}$) be the set of participants sending to (\resp receiving from) $p$.

Different semantics can be considered for the same network depending on the communication mechanism. A \emph{system} is a network together with a communication mechanism, denoted $\System{\Network}{com}{}$.  
It can communicate synchronously or asynchronously. 
In a synchronous system, each message sent is immediately received, \ie the communication exchange cannot be decoupled.
In an asynchronous communication instead, messages are stored in a memory. Here we only consider FIFO (First In First Out) buffers, which can be bounded or unbounded. Summing up, we deal with: 

\begin{itemize}
\item Synchronous (\synch): there is no buffer in the system, messages are immediately received when they are sent; 

\item \Ptp (\pp): there is a buffer for each pair of peers and direction of communication ($n \times (n - 1)$ buffers), where
one element of the pair is the sender and the other is the receiver; 

\item \MAILBOX (\mailbox): there are as many buffers as peers, each peer receives all its messages in a unique buffer, no matter the sender. 
\end{itemize}

We use \emph{configurations} to describe the state of a system and its buffers. 
\begin{definition}[Configuration]\label{def:configuration} 
Let $\Network = \left( (\Automaton{p})_{p \in \Part}, \messageSet\right)$ be a network.
  A \synch configuration (respectively a \pp configuration, or a \mailbox configuration)
          is a tuple $\config = \left( (\state{s}{}{p})_{p \in \Part}, \bufferSet\right)$ such that:
\begin{itemize}
\item $\state{s}{}{p} $ is a state of automaton $\Automaton{p}$, for all $p\in \Part$
\item $\bufferSet$ is a set of buffers whose content is a word over  $\messageSet$ with:
\begin{itemize}
\item an empty tuple for a \synch configuration, 
\item  a tuple  $(\buffer{12},\dots,\buffer{n(n-1)})$ for a \pp configuration, and
\item  a tuple $(\buffer{1},\dots,\buffer{n})$ for an \mailbox configuration. 
\end{itemize} 
 \end{itemize}
\end{definition}

We write $\emptyword$ to denote an empty buffer, and $\bufferSet^\emptyset$ to denote that all buffers are empty. We write  $\bufferSet\{b_i/b\}$ for the tuple of buffers $\bufferSet$, where  $\buffer{i}$ is substituted with $\buffer{}$.
We denote $\configSet$ the set of all configurations, 
$\configInit= \left((\state{s}{0}{p})_{p\in \Part}, \buffersInit \right)$ is the \emph{initial configuration}, and $\configFinal \subseteq \configSet$ is the set of \emph{final configurations}, where $\state{s}{}{p} \in \finalStates{p}$ for all participant $p \in \Part$.


We describe the behaviour of a system with \emph{runs}. A run is a sequence of transitions starting from an initial configuration $\configInit$.  
Let $\com \in \{\synch,~\pp,~\mailbox\}$ be the type of communication.
We define $\TransitionAsync{}{}{}{\com}{}^*$ as the transitive reflexive closure of $\TransitionAsync{}{}{}{\com}{}$. 

In order to simplify the definitions of executions and traces (given in what follows) and without loss of generality, we choose to label the transition with the sending message $\send{a}{p}{q}$. 

In a synchronous communication, we consider that the send and the receive of a message have done at the same time, \ie the synchronous relation \synch-send merges these two actions.

\begin{definition}[Synchronous system]\label{def:synchronousSystem}
 Let $\Network = \left((\Automaton{p})_{p\in \Part}, \messageSet\right)$ be a network. The synchronous system $\System{\Network}{\synch}{}$ associated with $\Network$ is the smallest binary relation $\xrightarrow[\synch]{}$ over \synch-configurations such that:
\begin{prooftree}\label{rule:0-send}
\AxiomC{$\Transition{\state{s}{}{p}}{\send{a}{p}{q}}{\state{s'}{}{p}}{p}$}
\AxiomC{$\Transition{\state{s}{}{q}}{\rec{a}{p}{q}}{\state{s'}{}{q}}{q}$}
\LeftLabel{(\synch-send)}
\BinaryInfC{$\TransitionSync{\left( (\state{s}{}{1}, \dots , \state{s}{}{p}, \dots , \state{s}{}{q}, \dots , \state{s}{}{n}), \bufferSet^\emptyset\right)}{\msg{a}{p}{q}}{\left( (\state{s}{}{1}, \dots , \state{s'}{}{p}, \dots , \state{s'}{}{q}, \dots , \state{s}{}{n}), \bufferSet^\emptyset\right)}$}
\end{prooftree}

\end{definition}

\begin{definition}[Peer-to-peer system]\label{def:peerToPeerSystem}
 Let $\Network = \left((\Automaton{p})_{p\in \Part}, \messageSet\right)$ be a network.
The peer-to-peer system $\System{\Network}{\pp}{}$ associated with $\Network$  is the least binary relation
  $\TransitionAsync{}{}{}{\pp}{}$ over $\pp$ configurations
  such that
for each configuration $C = \left( (\state{s}{}{p})_{p \in \Part}, \bufferSet\right)$, we have 
 $\bufferSet = (b_{pq})_{p \neq q \in \Part}$, 
  with
  $b_{pq}\in~\messageSet^*$,
and 
$\TransitionAsync{}{}{}{\pp}{}$ is the least transition induced by: 
\begin{prooftree}
\AxiomC{$\Transition{\state{s}{}{p}}{\send{a}{p}{q}}{\state{s'}{}{p}}{p}$}
\LeftLabel{(\pp-send)}
\UnaryInfC{$\TransitionAsync{\left((\state{s}{}{1}, \dots , \state{s}{}{p}, \dots , \state{s}{}{n}), \bufferSet\right)}{\send{a}{p}{q}}{\left((\state{s}{}{1}, \dots , \state{s'}{}{p}, \dots , \state{s}{}{n}), \bufferSet\{b_{pq} / b_{pq}\cdot a\}\right)}{\pp}{}$}
\end{prooftree}

\begin{prooftree}
\AxiomC{$\Transition{\state{s}{}{q}}{\rec{a}{p}{q}}{\state{s'}{}{q}}{q}$}
\AxiomC{$b_{pq} = a\cdot b_{pq}'$ }
\LeftLabel{(\pp-rec)}
\BinaryInfC{
$\TransitionAsync{\left(( \state{s}{}{1}, \dots , \state{s}{}{q}, \dots , \state{s}{}{n}), \bufferSet\right)}{\rec{a}{p}{q}}{\left((\state{s}{}{1}, \dots , \state{s'}{}{q}, \dots ,\state{s}{}{n}), \bufferSet\{b_{pq} / b_{pq}'\}\right)}{\pp}{}$}
\end{prooftree}
\end{definition}

\begin{definition}[Mailbox system]\label{def:mailboxSystem}
Let $\Network = \left((\Automaton{p})_{p\in \Part}, \messageSet\right)$ be a network.
The mailbox system $\System{N}{\mailbox}{}$ associated with $\Network$  is the smallest binary relation
  $\TransitionAsync{}{}{}{\mailbox}{}$ over $\mailbox$-configurations such that 
for each configuration $C = \left( (\state{s}{}{p})_{p \in \Part}, \bufferSet\right)$, we have  
 $\bufferSet= (b_{p})_{p \in \Part}$ 
and 
 $\TransitionAsync{}{}{}{\mailbox}{}$ is the smallest transition such that: 
\begin{prooftree}
\AxiomC{$\Transition{\state{s}{}{p}}{\send{a}{p}{q}}{\state{s'}{}{p}}{p}$}
\LeftLabel{(\mailbox-send)}
\UnaryInfC{
$\TransitionAsync
	{\left((\state{s}{}{1}, \dots ,\state{s}{}{p}, \dots ,\state{s}{}{n}), \bufferSet\right)}
	{\send{a}{p}{q}}
	{\left((\state{s}{}{1}, \dots ,\state{s'}{}{p}, \dots , \state{s}{}{n}), \bufferSet\{b_{q} / b_{q}\cdot a\}\right)}{\mailbox}{}$}
\end{prooftree}

\begin{prooftree}
\AxiomC{$\Transition{\state{s}{}{q}}{\rec{a}{p}{q}}{\state{s'}{}{q}}{q}$}\AxiomC{$b_{q} = a\cdot b_{q}'$ }
\LeftLabel{(\mailbox-rec)}
\BinaryInfC{$\TransitionAsync{\left((\state{s}{}{1}, \dots ,\state{s}{}{q}, \dots ,\state{s}{}{n}), \bufferSet\right)}{\rec{a}{p}{q}}{\left((\state{s}{}{1}, \dots , \state{s'}{}{q}, \dots , \state{s}{}{n}), \bufferSet \{b_{q} / b_{q}'\}\right)}{\mailbox}{}$}
\end{prooftree}
\end{definition}


In order to study the behaviour of systems, we define the set of executions and traces. An execution $\exec$ is a sequence of actions leading to a final global state and the corresponding trace $\trace$ is the projection on the send actions\footnote{In Definition~\ref{def:traces}, we decide not to take into consideration the content of buffers, differently to others papers, like \cite{finkel_synchronizability_2017}, where the authors study stable configurations, \ie configurations where buffers are empty.}.

\begin{definition}[Execution]\label{def:execution}
  Let $\Network = \left((\Automaton{p})_{p\in \Part}, \messageSet\right)$ be a network and  $\com\in \{\synch, \pp, \mailbox\}$  
 be the type of communication.  
  $\execSet{\System{\Network}{\com}{}}$ is the set of executions defined, with 
  $\configInit$ the initial configuration, $\config_n$ a final configuration and $\action_i \in 
  \actionSet{}$ for all $1 \leq i \leq n$, by:
$$\execSet{\System{\Network}{\com}{}} = \{\action_1\cdot\ \dots\ \cdot \action_n \mid  \TransitionAsync{\configInit}{\action_1}{\config_1}{\com}{}{} \TransitionAsync{}{\action_2}{}{\com}{}{} \dots \TransitionAsync{}{\action_n}{\config_n}{\com}{}{}\}.$$  
\end{definition}






If $\word$ is a word over actions, then let $ \word\projOut $ (\resp $ \word\projIn $) be its projection on only send (\resp receive) actions, let $ \word\projPeers{P} $ its projection on only actions that involve only the participants in a set $ P $, let $ \word\projPeer{p} $ its projection on receives towards $ p $ and sends from $ p $, and let $ \word\projMess $ be the word over messages that results from $ w $ by removing all $ ! $ and $ ? $.
We extend the operators $ \projOut $, $ \projIn $, $ \projPeers{P} $, $ \projPeer{p} $, and $ \projMess $ to languages, by applying them on every word of the language.
Note that, $ \word\projPeers{\Set{p}} $ is always empty, since there are no actions that involve only a single participant $ p $, whereas $ \word\projPeer{p} $ is the projection of $ w $ to its actions in that $ p $ has an active role (sender in send actions and receiver in receive actions).

\begin{definition}[Traces]\label{def:traces}
	Let $\Network = \left((\Automaton{p})_{p\in \Part}, \messageSet \right)$ be a network and  $\com\in \{\synch,~ \pp,~\mailbox\}$ be the type of communication.
	$\traceSet{\System{\Network}{\com}{}}$ is the set of traces: $$\traceSet{\System{\Network}{\com}{}} = \{\exec\projOut \mid \exec \in \execSet{\System{\Network}{\com}{}}\}.$$	
\end{definition}

A system is synchronisable if its asynchronous behaviour can be related to its synchronous one. Thus, an asynchronous system is synchronisable if its set of traces is the same as the one obtained from the synchronous system. 

\begin{definition}[Synchronisability]\label{def:synchronisability}
Let $\Network = \left((\Automaton{p})_{p\in \Part}, \messageSet\right)$ be a network and  $\com\in \{\pp,~\mailbox\}$ be the type of communication. The system $\System{\Network}{\com}{}$ is synchronisable if and only if 
$\traceSet{\System{\Network}{\com}{}} = \traceSet{\System{\Network}{\synch}{}}$. 
\end{definition} 

\paragraph*{Problems statements.} 
We define the \emph{Synchronisability Problem} as the decision problem of determining whether a given system, where all states are accepting states, is synchronisable or not. 
We also consider the \emph{Generalised Synchronisability Problem} without any constraints on the accepting states of the system.


\section{The Generalised Synchronisability Problem is Undecidable}\label{sec:post}

The first contribution is about assessing the undecidability of the Generalised Synchronisability Problem for the \Mailbox semantics. 
This result strongly relies on the notion of accepting word.
Moreover, the entire section considers networks without any constraints on final configurations (\ie $\configFinal \subseteq \configSet$).  
\paragraph{Post Correspondence Problem.}
We will resort to the Post Correspondence Problem (PCP) which is known to be an undecidable decision problem~\cite{post_variant_1946}, to prove that the Generalised Synchronisability Problem is undecidable. 
We will show that the encoding of a PCP instance \pcpinstance{} is not synchronisable if and only if the instance has a solution.



\begin{definition}[Post Correspondence Problem]\label{def:pcp}
Let $\alphabet$ be an alphabet with at least two symbols. An instance \pcpinstance{} of the PCP consists of two finite ordered lists of the same number of non-empty words $$\Word= {w_{1}, w_{2},\ldots , w_{n}}\text{ and }\text{\Word'}= {w'_{1}, w'_{2}, \ldots , w'_{n}}$$ such that $w_{i}, w'_{i} \in \alphabet^{*}$ for all indices $1 \leq i \leq n$. 
A solution of this instance is a finite sequence of indices $\Sol = ( i_{1}, i_{2}, \dots, i_{m} )$ with $m \geq 1$ and $i_{j} \in [1,n]$ for all $1 \leq j \leq m$ such that:
$$w_{i_{1}} \cdot w_{i_{2}} \cdot \ldots  \cdot w_{i_{m}} = w'_{i_{1}} \cdot w'_{i_{2}} \cdot \ldots  \cdot w'_{i_{m}} .$$
\end{definition}

\paragraph{Mailbox Encoding of the Post Correspondence Problem.}
The encoding in \Mailbox systems requires some care.
When an automaton is receiving messages from multiple participants, these messages are interleaved in the buffer and it is generally not possible to anticipate in which order these messages have been sent.

The encoding of an instance \pcpinstance{} of the PCP  is a parallel composition of four automata: $\Automaton{I}$, $\Automaton{W}$, $\Automaton{W'}$, and $\Automaton{L}$, where $\Automaton{I}$ sends the same indices to $\Automaton{W}$ and $\Automaton{W'}$ which in turn send the respective words to $\Automaton{L}$. $\Automaton{L}$ compare letters and, at the end of the run, its state allows to say if a solution exists. The topology and the buffer layout of the system is depicted in Figure~\ref{figure:topologyEncodeMail}. 

\begin{figure}[t]
	\centering
	\scalebox{0.9}{\begin{tikzpicture}[>=stealth,node distance=1.5cm,shorten >=1pt, every state/.style={text=black},semithick]
	\node (ai) at (-2,0) {$\Automaton{I}$};
	\node (aw) at (0,0.75) {$\Automaton{W}$};
	\node (awp) at (0,-0.75) {$\Automaton{W'}$};
	\node (al) at (2,0) {$\Automaton{L}$};
	\draw[->] (ai) -- (aw);
	\draw[->] (ai) -- (awp);
	\draw[->] (aw) -- (al);
	\draw[->] (awp) -- (al);
	\draw[->] (al) -- (ai);
\end{tikzpicture}}
	\caption{Topology of the encoding of an instance  \pcpinstance{} of the PCP} 
	\label{figure:topologyEncodeMail}
\end{figure}
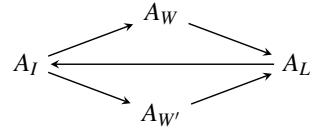

We will explain our encoding over an example. Take the following PCP instance with  $\alphabet = \{a, b\}$,  $\Word = a, b, abab$ and $\Word' = ba, baa, b$. We know that there is a solution for this instance with $\Sol = (2, 1, 3)$. 
 Figures~\ref{figure:example:encodingMailAutBI}--\ref{figure:example:encodingMailAutBL} 
depict the automata solving the PCP instance.


Automaton $\Automaton{I}$ guesses the sequence of indexes and sends it to both $\Automaton{W}$ and $\Automaton{W'}$. The message $\$$ is used to signal the end of the sequence.
Automaton $\Automaton{W}$ and  $\Automaton{W'}$ receive indexes  from $\Automaton{I}$ and send the corresponding sequences of letters to $\Automaton{L}$. At the reception of message $\$$, they send messages $end$ to $\Automaton{L}$.
Automaton $\Automaton{L}$ checks whether the sequences of letters produced by $\Automaton{W}$ and $\Automaton{W'}$ coincide. Letters from  $\Automaton{W}$ and $\Automaton{W'}$ need to be alternate and are read in turn, and the additional receptions are used to make the system synchronisable and to recognize errors (\ie sequences that are not a solution).
If all comparisons succeed, included the  $end$  messages, then $\Automaton{L}$  sends message $ok$ that is not received by any participant and ends up in the unique accepting state.

More formally, we define the encoding as follows. 
\begin{definition}[Encoding of PCP in mailbox system]\label{def:encodingMail}
 Let $(\Word,\Word')$ be a PCP instance over $\alphabet$. 
The encoding of $(\Word,\Word')$ is the network $\Encode{\Word,\Word'}{\mailbox} =  \left((\Automaton{p})_{p \in \Part}, \messageSet\right)$ where:
\begin{itemize}
	\item $\Part = \left\{ I, W, W', L\right\} $
	\item $\messageSet = \left\{ \msg{i}{I}{W}, \msg{i}{I}{W'} \mid i \in [1,n]\right\} \cup \left\{ \msg{\alpha}{W}{L},\msg{\alpha}{W'}{L} \mid \alpha \in \Sigma \right\} \cup M $ with\\
		$ M = \left\{ \msg{\$}{I}{W}, \msg{\$}{I}{W'}, \msg{end}{W}{L}, \msg{end}{W'}{L}, \msg{ok}{L}{I}\right\} $
\begin{figure}[t]
	\centering
	\begin{tikzpicture}[>=stealth,node distance=1.5cm,shorten >=1pt,
    every state/.style={minimum size = 0pt},semithick]

      \node[state,initial,initial text={}]   at (0,0) (q0)  {$0$};
      \node[state] at (-1.5,1.7) (q1)  {$1$};
      \node[state] at (-1.5,-1.7) (q2)  {$2$};
      \node[state] at (1.5,-1.7) (q3)  {$3$};
      \node[state] at (2, 0) (qd)  {$\$$};
      \node[state, accepting] at (4, 0) (qd')  {$\$'$};

	\path (q0) edge[->, bend right = 10] node [sloped, above] {$\send{1}{I}{W}$} (q1);
	\path (q1) edge[->, bend right = 10] node [sloped, below] {$\send{1}{I}{W'}$} (q0);    
	
	\path (q0) edge[->, bend left=10] node [sloped, below] {$\send{2}{I}{W}$} (q2);
	\path (q2) edge[->, bend left=10] node [sloped, above] {$\send{2}{I}{W'}$} (q0);    
	
	\path (q0) edge[->, bend left=10] node [above,sloped] {$\send{3}{I}{W}$} (q3);
	\path (q3) edge[->, bend left=10] node [below, sloped] {$\send{3}{I}{W'}$} (q0);    
	
	\path (q0) edge[->] node [above] {$\send{\$}{I}{W}$} (qd);
	\path (qd) edge[->] node [above] {$\send{\$}{I}{W'}$} (qd');

\end{tikzpicture} 
	\caption{Automaton $\Automaton{I}$} 
	\label{figure:example:encodingMailAutBI}
\end{figure}
	\item  $\Automaton{I} = (\stateSet{I}, \state{s}{0}{I}, \messageSet, \rightarrow_I, \finalStates{\Automaton{I}} )$ where $\stateSet{I} = \left\{\state{q}{0}{}, \state{q}{\$}{}, \state{q}{\$'}{}\right\} \cup \left\{ \state{q}{i}{} \mid i \in [1,n] \right\}$, $\state{s}{0}{I} = \state{q}{0}{}$, $\finalStates{\Automaton{I}} = \Set{\state{q}{\$'}{}}$ and 
		$$ 
		\rightarrow_I  = \ \left\{ \TransitionAsync
					{\state{q}{0}{}}
					{\send{i}{I}{W}}
					{\state{q}{i}{}}{}{},
					\TransitionAsync	
					{\state{q}{i}{}}
					{\send{i}{I}{W'}}
					{\state{q}{0}{}}{}{}			
					\mid i \in [1,n]\right\} 
		 \cup \left\{ \TransitionAsync
					{\state{q}{0}{}}
					{\send{\$}{I}{W}}
					{\state{q}{\$}{}}{}{},
		 \TransitionAsync
					{\state{q}{\$}{}}
					{\send{\$}{I}{W'}}
					{\state{q}{\$'}{}}{}{} \right\} $$
\begin{figure}[t]
	\centering
	\begin{tikzpicture}[>=stealth,node distance=1.5cm,shorten >=1pt,
  every state/.style={minimum size = 0pt},semithick]

      \node[state,initial,initial text={}]   at (0.2,0) (q0)  {$0$};
      \node[state] at (2,2) (q10)  {$1,0$};
      \node[state] at (2,0) (q20)  {$2,0$};
      \node[state] at (2,-2) (q30)  {$3,0$};
      \node[state] at (4,-2) (q31)  {$3,1$};
      \node[state] at (6,-2) (q32)  {$3,2$};
      \node[state] at (8,-2) (q33)  {$3,3$};
      \node[state] at (6,0) (qf)  {$f$};
      \node[state] at (6,2) (qd)  {$\$$};
      \node[state, accepting] at (8,2) (qe)  {$e$};

    \path (q0) edge[->, bend left] node [left] {$\rec{1}{I}{W}$} (q10);
    \path (q0) edge[->] node [above] {$\rec{2}{I}{W}$} (q20);
    \path (q0) edge[->, bend right] node [left] {$\rec{3}{I}{W}$} (q30);

    \path (q10) edge[->, bend left=25] node [sloped, above] {$\send{a}{W}{L}$} (qf);
    \path (qf) edge[->] node [sloped, above] {$\rec{1}{I}{W}$} (q10);

    \path (q20) edge[->, bend left=10] node [sloped, above] {$\send{b}{W}{L}$} (qf);
    \path (qf) edge[->, bend left=10] node [sloped, below] {$\rec{2}{I}{W}$} (q20);
    
    \path (q30) edge[->] node [below] {$\send{a}{W}{L}$} (q31);
    \path (q31) edge[->] node [below] {$\send{b}{W}{L}$} (q32);
    \path (q32) edge[->] node [below] {$\send{a}{W}{L}$} (q33);
    \path (q33) edge[->, bend right] node [sloped, above] {$\send{b}{W}{L}$} (qf);
    \path (qf) edge[->, bend right=5] node [pos = .3, sloped, below] {$\rec{3}{I}{W}$} (q30);
    
    \path (qf) edge[->] node [right] {$\rec{\$}{I}{W}$} (qd);
    \path (qd) edge[->] node [above] {$\send{end}{W}{L}$} (qe);

\end{tikzpicture} 
	\caption{Automaton $\Automaton{W}$} 
	\label{figure:example:encodingMailAutBW}
\end{figure}
\begin{figure}[t]
	\centering
	\begin{tikzpicture}[>=stealth,node distance=1.5cm,shorten >=1pt,
  every state/.style={minimum size = 0pt},semithick]

      \node[state,initial,initial text={}]   at (0.1,0) (q0)  {$0$};
      \node[state] at (2,2) (q10)  {$1,0$};
      \node[state] at (4.1, 2) (q11) {$1,1$};
      \node[state] at (2,-2) (q20)  {$2,0$};
      \node[state] at (4.1, -2) (q21) {$2,1$};
      \node[state] at (6.5, -2) (q22) {$2,2$};
      \node[state] at (2,0)  (q30)  {$3,0$};
      \node[state] at (6.5,0) (qf)  {$f$};
      \node[state] at (6.5,2) (qd)  {$\$$};
      \node[state, accepting] at (8.7,2) (qe)  {$e$};

    \path (q0) edge[->, bend left] node [left] {$\rec{1}{I}{W'}$} (q10);
    \path (q0) edge[->, bend right] node [left] {$\rec{2}{I}{W'}$} (q20);
    \path (q0) edge[->] node [above] {$\rec{3}{I}{W'}$} (q30);

    \path (q10) edge[->] node [above] {$\send{b}{W'}{L}$} (q11);
    \path (q11) edge[->, bend left = 15] node [above, sloped] {$\send{a}{W'}{L}$} (qf);
    \path (qf) edge[->, bend left = 5] node [above, sloped] {$\rec{1}{I}{W'}$} (q10);
    
    \path (q20) edge[->] node [sloped, below] {$\send{b}{W'}{L}$} (q21);
    \path (q21) edge[->] node [below] {$\send{a}{W'}{L}$} (q22);
    \path (q22) edge[->] node [right] {$\send{a}{W'}{L}$} (qf);
    \path (qf) edge[->, bend right=5] node [below, sloped] {$\rec{2}{I}{W'}$} (q20);

    \path (q30) edge[->, bend left = 5] node [above, sloped, pos= .4] {$\send{b}{W'}{L}$} (qf);
    \path (qf) edge[->, bend left = 5] node [below, sloped, pos = .6] {$\rec{3}{I}{W'}$} (q30);
    
    \path (qf) edge[->] node [right] {$\rec{\$}{I}{W'}$} (qd);
    \path (qd) edge[->] node [above] {$\send{end}{W'}{L}$} (qe);

\end{tikzpicture} 
	\caption{Automaton $\Automaton{W'}$} 
	\label{figure:example:encodingMailAutBWb}
\end{figure}
	\item  $\Automaton{W} = (S_W, \state{s}{0}{W}, \messageSet, \rightarrow_W, \finalStates{\Automaton{W}})$ where 
$S_W  = \left\{\state{q}{0}{}, \state{q}{f}{}, \state{q}{\$}{}, \state{q}{e}{}\right\} \cup\left\{q_{i,j} \mid i \in [1,n] \wedge j \in [0, \length{w_i}-1 ] \right\}$,\linebreak $\state{s}{0}{W}  = \state{q}{0}{}$, $\finalStates{\Automaton{W}} = \Set{\state{q}{e}{}}$ and 
		\begin{align*}
		 \rightarrow_W ={}& \left\{ \Transition
				{\state{q}{0}{}}
				{\rec{i}{I}{W}}
				{\state{q}{i,0}{}}
				{},
				\Transition
				{\state{q}{f}{}}
				{\rec{i}{I}{W}}
				{\state{q}{i,0}{}}
				{}
				\mid i \in [1,n] \right\}  
		 \cup \left\{ \Transition
				{\state{q}{i,|w_i|-1}{}}
				{\send{\alpha}{W}{L}}
				{\state{q}{f}{}}
				{} \mid \alpha = w_{i,|w_i|} \wedge i \in [1,n] \right\}  \\ 
		&{}\cup \left\{ \Transition
				{\state{q}{i,j}{}}
				{\send{\alpha}{W}{L}}
				{\state{q}{i,j+1}{}}
				{}\mid   \alpha = w_{i,j+1} \wedge i \in [1,n] \wedge j \in [1, |w_i|-2]  \right\}\\ 
		&{}\cup \left\{ \Transition
				{\state{q}{f}{}}
				{\rec{\$}{I}{W}}
				{\state{q}{\$}{}}
				{}, \Transition
				{\state{q}{\$}{}}
				{\send{end}{W}{L}}
				{\state{q}{e}{}}
				{} \right\}
		\end{align*}
\begin{figure}[t]
	\centering
	\begin{tikzpicture}[>=stealth,node distance=1.5cm,shorten >=1pt,
    every state/.style={minimum size = 0pt}, semithick]

      \node[state,initial,initial text={}]   at (0,0) (q0)  {$0$};
      \node[state] at (2.5,1.5) (qa)  {$a$};
      \node[state] at (2.5,-1.5) (qb)  {$b$};
      \node[state] at (7,0) (q*)  {$*$};
      \node[state] at (3.5, 3.5) (qe)  {$e$};
      \node[state] at (6, 4) (qe')  {$e'$};
      \node[state, accepting] at (8, 4) (qok)  {$ok$};
      
	\path (q0) edge[->, bend left=10] node [above, sloped] {$\rec{a}{W}{L}$} (qa);
	\path (qa) edge[->, bend left=10] node [below, sloped] {$\rec{a}{W'}{L}$} (q0);    
	
	\path (q0) edge[->, bend left=10] node [above, sloped, pos = .6] {$\rec{b}{W}{L}$} (qb);
	\path (qb) edge[->, bend left=10] node [below, sloped] {$\rec{b}{W'}{L}$} (q0);    
	
	\path (q0) edge[->, bend left=40] node [left] {$\rec{end}{W}{L}$} (qe);
	\path (qe) edge[->] node [above] {$\rec{end}{W'}{L}$} (qe');
	\path (qe') edge[->] node [above] {$\send{ok}{L}{I}$} (qok);
	
	\path (q0) edge[->, bend right =90] node [below, scale = 0.9] 
		{\begin{tabular}{l l } 			
				$\rec{a}{W'}{L}$ & \\ 
				$\rec{b}{W'}{L}$ & $\rec{end}{W'}{L}$ \\ 
		\end{tabular}}	
	(q*); 
	
	\path (qb) edge[->, bend right =20] node [above, scale = 0.9, pos = .4] 
		{\begin{tabular}{l l} 
				$\rec{a}{W}{L}$ & $\rec{a}{W'}{L}$ \\ 
				$\rec{b}{W}{L}$  & $\rec{end}{W'}{L}$\\ 
				$\rec{end}{W}{L}$  \\ 
		\end{tabular}}	
	(q*); 
	
	\path (qa) edge[->, bend left =20] node [above, scale = 0.9] 
		{\begin{tabular}{l l} 
				$\rec{a}{W}{L}$ &   \\ 
				$\rec{b}{W}{L}$ &  $\rec{b}{W'}{L}$ \\ 
				$\rec{end}{W}{L}$ & $\rec{end}{W'}{L}$\\	
		\end{tabular}}	
	(q*); 
	
	\path (qe) edge[->, bend left =40] node [right, scale = 0.9] 
		{\begin{tabular}{l} 
				$\rec{a}{W'}{L}$  \\
				$\rec{b}{W'}{L}$  \\
				 \\
		\end{tabular}}	
	(q*); 
	
	\path (q*) edge[->, loop right ] node [right, scale = 0.9] {$*$}
	();

\end{tikzpicture} 
	\caption{Automaton $\Automaton{L}$} 
	\label{figure:example:encodingMailAutBL} 
\end{figure}
	\item $\Automaton{L} = (S_L, \state{s}{0}{L}, M, \rightarrow_L, \finalStates{\Automaton{L}})$ where $S_L = \{\state{q}{0}{}, \state{q}{e}{}, \state{q}{e'}{}, \state{q}{ok}{}, \state{q}{*}{}\} \cup \{\state{q}{\alpha}{} | \alpha \in \Sigma \}$, $\state{s}{0}{L} = \state{q}{0}{}$, $\finalStates{\Automaton{L}} = \Set{\state{q}{ok}{}}$ and 
		\begin{align*} \rightarrow_L =
		&\ \left\{ \Transition
				{\state{q}{0}{}}
				{\rec{\alpha}{W}{L}}
				{\state{q}{\alpha}{}}
				{},
				\Transition
				{\state{q}{\alpha}{}}
				{\rec{\alpha}{W'}{L}}
				{\state{q}{0}{}}
				{}\mid \alpha \in \Sigma \right\}   
		 \cup \left\{ \Transition
				{\state{q}{\alpha}{}}
				{\rec{\beta}{W'}{L}}
				{\state{q}{*}{}}
				{}\mid \beta \in \Sigma \cup \{end\} \wedge \beta \neq \alpha \right\} \\%
		&\ \cup \left\{ 
				 \Transition
				{\state{q}{0}{}}
				{\rec{\alpha}{W'}{L}}
				{\state{q}{*}{}}
				{} \mid \alpha \in \Sigma \cup \{end\}\right\} \cup
				\left\{
				\Transition
				{\state{q}{\alpha}{}}
				{\rec{\beta}{W}{L}}
				{\state{q}{*}{}}
				{} \mid \beta \in \Sigma \cup \{end\} \right\} \\ 
		&\ \cup \left\{ \Transition
				{\state{q}{e}{}}
				{\rec{\alpha}{W}{L}}
				{\state{q}{*}{}}
				{}\ |\ \alpha \in \Sigma \right\} \cup \left\{ \Transition
				{\state{q}{e}{}}
				{\rec{\alpha}{W'}{L}}
				{\state{q}{*}{}}
				{}\ |\ \alpha \in \Sigma \right\} \\ 
		&\ \cup \left\{ \Transition
				{\state{q}{*}{}}
				{\rec{\alpha}{X}{L}}
				{\state{q}{*}{}}
				{}\ |\ \alpha \in \Sigma \cup \{end\} \wedge X \in \{W, W'\}\right\} \\ 
		&\ \cup \left\{  
				\Transition
				{\state{q}{0}{}}
				{\rec{end}{W}{L}}
				{\state{q}{e}{}}
				{},
				 \Transition
				{\state{q}{e}{}}
				{\rec{end}{W'}{I}}
				{\state{q}{e'}{}}
				{}, 
				\Transition
				{\state{q}{e'}{}}
				{\send{ok}{L}{I}}
				{\state{q}{ok}{}}
				{} \right\}
		\end{align*}
\end{itemize}
$\Automaton{W'}$ is defined as $\Automaton{W}$ but considering $\Word'$ instead of $\Word$. 
\end{definition}

It is easy to see that in the synchronous semantics the system cannot reach any final configuration, because of message $ok$ which cannot be sent since it cannot be received. The set of traces of the synchronous system is indeed empty. 

 \begin{lemma}  \label{lem:emptyMail}
Let \pcpinstance{} an instance of PCP and $\Network  = \Encode{\Word,\Word'}{\mailbox}$ its encoding into communicating automata. Then $\traceSet{\System{\Network}{\synch}{}} = \emptyset$.
 \end{lemma} 

In the \Mailbox semantics, message $ok$ can be sent only if the encoded instance of PCP has a solution.
If the instance of PCP has no solution, then the \Mailbox system is unable to reach the final configuration and the set of traces is empty.
Summing up, the set of traces is not empty if and  only if there exists a solution to the corresponding PCP instance.

\begin{lemma} \label{lem:noEmptyMail}
For every instance \pcpinstance{} of PCP, where $\Network  = \Encode{\Word,\Word'}{\mailbox}$, \pcpinstance{} has a solution if and only if $\traceSet{\System{N}{\mailbox}{}}  \neq \emptyset$.
\end{lemma}

\begin{proof}
Let \pcpinstance{} be a PCP instance and $\Network  = \Encode{\Word,\Word'}{\mailbox}$.
\begin{compactdesc}
\item[$\Rightarrow$] We show that if \pcpinstance{} has a solution, then 
 $\traceSet{\System{N}{\mailbox}{}}  \neq \emptyset$.
 Let $\Sol_{\pcpinstance{}} = (i_1, i_2, \ldots ,  i_m)$ be a solution of \pcpinstance{}.
 Let $w= a_1 \dots a_n$ be the word generated from the sequence of indices.
 From Definition~\ref{def:encodingMail}, it is easy to see that the following execution $t$ is possible and that it leads to a final configuration with the final global state $(\state{q}{\$'}{I}, \state{q}{e}{W},\state{q}{e}{W'}, \state{q}{ok}{L})$:

\begin{align}
t = & \send{i_1}{I}{W} \cdot \send{i_1}{I}{W'}\cdot 
\ldots \cdot 
\send{i_m}{I}{W}\cdot \send{i_m}{I}{W'}
\cdot
\send{\$}{I}{W}\cdot \send{\$'}{I}{W'} \label{one}\\
&\cdot
\send{a_1}{W}{L}\cdot \send{a_1'}{W'}{L}
\cdot \ldots \cdot 
\send{a_n}{W}{L}\cdot \send{a_n'}{W'}{L}
\cdot
\send{end}{W}{L}\cdot \send{end}{W'}{L} \label{two}\\
&\cdot 
\send{ok}{L}{I} \label{three}
\end{align}
Part~(\ref{one}) consists of the indices sent by automaton 
\Automaton{I} in turn to the automata \Automaton{W} and \Automaton{W'}, including the messages $\$, \$'$ that are used to signal the end of the sequence.
Part~(\ref{two}) contains the letters of word $w$ sent in turn by  
\Automaton{W} and \Automaton{W'} upon reception of the corresponding indices to \Automaton{L}. Since we are considering mailbox communication here, note that messages from \Automaton{W} and \Automaton{W'} must alternate.  
Finally, automaton \Automaton{L} having matched all the words from \Automaton{W} and \Automaton{W'}, including the final $end$ messages is able to send the last message $ok$, part~(\ref{three}).
Hence  $t \in \traceSet{\System{N}{\mailbox}{}}$.
\item[$\Leftarrow$] 
Conversely, we show that if $t \in \traceSet{\System{N}{\mailbox}{}}$, then there is a solution to \pcpinstance{}.
Since $t \in \traceSet{\System{N}{\mailbox}{}}$, $t$ is the projection on send messages of an accepting execution $t' \in \execSet{\System{\Network}{\mailbox}{}}$.
By construction, to reach state $\state{q}{ok}{L}$ 
 we know that $t' = t_1 \cdot \rec{end}{W}{L}\cdot \rec{end}{W'}{L} \cdot \send{ok}{L}{I} \send{ok}{L}{I}$. 
 With a similar reasoning, to reach states $\state{q}{e}{W}$ and $\state{q}{e}{W'}$,  $t_1\projOut = t_2\projOut  \cdot \send{end}{W}{L} \cdot \send{end}{W'}{L}$.
 This also entails that there has been at least one index sent by automaton \Automaton{I} (both to \Automaton{W} and \Automaton{W'}).
 In turn, upon reception of the corresponding index,  \Automaton{W} and \Automaton{W'} send the corresponding letters to \Automaton{L}. 
 The sequence can only be accepted if letters are queued in order: one letter from  \Automaton{W} followed by the same letter from \Automaton{W'}.
Hence if we take the projection of $t$ on the actions of \Automaton{I} we obtain a sequence of indices that represent a solution to \pcpinstance{}. 
\qedhere
 \end{compactdesc}
\end{proof}

Therefore, the system is synchronisable if and only if the encoded instance does not have solution.  
%
%
%
\begin{theorem}\label{theorem:mail}
The Generalised Synchronisability Problem is undecidable for mailbox systems. 
\end{theorem}
\begin{proof}
Let \pcpinstance{} be an instance of PCP. 
\begin{compactdesc}
\item[$\Rightarrow$] 
If \pcpinstance{} has a solution, then  by Lemma~\ref{lem:noEmptyMail} 
$\traceSet{\System{N}{\mailbox}{}} \neq \emptyset$ and 
by Lemma~\ref{lem:emptyMail} 
$\traceSet{\System{N}{\synch}{}} = \emptyset$.
Hence, the system is not synchronisable.
\item[$\Leftarrow$] 
Conversely, if \pcpinstance{} has no solution, then by Lemma \ref{lem:noEmptyMail} $\traceSet{\System{N}{\mailbox}{}} = \emptyset$ and $\traceSet{\System{N}{\synch}{}} = \emptyset$ by Lemma~\ref{lem:emptyMail}. Hence the system is synchronisable.
\qedhere
\end{compactdesc} 
\end{proof}

%
%
%
%
%
%


\section{Synchronisability of \MAILBOX Communication for Tree-like Topologies}
\label{sec:decidability}
We are interested in the Synchronisability Problem, where automata have no final states. Notice that this is equivalent of having automata where all states are final. Thus, we  consider networks where all configurations are final configurations \mbox{($\configFinal = \configSet$).} 
The encoding in Section \ref{sec:post} cannot be used as it strongly relies on the existence of special final configurations that can only be reached in the asynchronous (\Mailbox) semantics. Moreover, because of the nature of \Mailbox communications, the encoding in \cite{finkel_synchronizability_2017} cannot be used. In fact, the order of messages received from different recipients becomes important and the relative speeds of the automata ($W$ and $W'$) producing the letters to be compared, cannot be ``synchronised".

In order to understand the expressiveness of \Mailbox system, we start by constraining the shape of topologies. A topology (cfr. Definition \ref{def:topology}) is the underlining communication structure marking the direction of communication among participants.
Here  we start by considering topologies that form a tree and we want to understand whether the topology impacts (or not) the decidability of the Synchronisability Problem. 
When considering tree topologies, each of the inner automata (nodes) receives messages by only one other participant. Because of this, systems with  tree topologies will have the same set of executions for both  \Mailbox and \Ptp semantics.

\begin{definition}[Tree topology]
	Let $\Network=((\Automaton{p})_{p\in \Part}, \messageSet)$ be a network of communicating automata and $\topo{\Network} = (\vertexSet,\edgeSet)$ its topology. 
	$\topo{\Network} = (\vertexSet,\edgeSet)$ is a tree if it is connected, without any cycle, and $ \mid \senders{p}\mid \leq 1 $ for all $ p \in \Part $.
\end{definition}

Let $ \rootP \in \Part $ denote the root of the tree, \ie $ \senders{\rootP} $ is empty. Notice that
$ \senders{p} $ is a singleton for all inner nodes $ p\in \Part\setminus\Set{\rootP} $.

It is interesting to see that we can characterise an algorithm to check whether a system is synchronisable or not.
To this aim, a system needs to validate two conditions: 
\begin{enumerate}
	\item the automata should provide matching receptions whenever their communication partners are ready to send and
	\item for each send of a parent there is a matching reception of the child. 
\end{enumerate}
The main idea is to use the tree structure to capture the influence the automata have on the language of each other.
The receptions of an automaton depend only on the availability of matching incoming messages, \ie in a tree by the sends of at most one parent.
We compute the \emph{influenced language} of $ \Automaton{p} $, denoted $ \LangTree{\Automaton{p}} $, considering only the influence of its parent but not of its children.
These languages have to be computed from the root $ \rootP $---that does not depend on anybody---towards the leafs of the tree.
For an inner node of the tree the possible sequences of outputs of the respective (unique) parent node determine the possible sequences of inputs it can perform and thus the outputs that can be unguarded.
The languages $ \LangTreeIn{\Automaton{p}} $ and $ \LangTreeOut{\Automaton{p}} $ are its respective projections on only receives or sends.

\begin{definition}[Influenced languages]
	Let $p \in \Part$.
	We define the influenced language as follows:
	\begin{align*}
		\LangTree{\Automaton{p}} &=
			\begin{cases}
				\languageof{\Automaton{\rootP}} & \text{if } p = \rootP\\
				\Set{w \mid w \in \languageof{\Automaton{p}} \wedge \left( w\projIn \right)\projMess \in \left( \LangTreeOut{\Automaton{q}} \right)\projMess \wedge \senders{p} = \Set{q} } & \text{otherwise}
			\end{cases}\\
		\LangTreeIn{\Automaton{p}} &= \LangTree{\Automaton{p}}\projIn\\
		\LangTreeOut{\Automaton{p}} &= \LangTree{\Automaton{p}}\projOut
	\end{align*}
\end{definition}

Since the root does not receive any message, it is not influenced by any parent.
Hence, $ \LangTreeIn{\Automaton{\rootP}} = \Set{\varepsilon} $ and $ \LangTree{\Automaton{\rootP}} = \LangTreeOut{\Automaton{\rootP}} = \languageof{\Automaton{\rootP}} $.
For any inner node $ p \in \Part $ of the tree, we allow only  words with input sequences that match  a sequence of outputs of its parent $ q $ influenced language.
To match inputs with their corresponding outputs, we ignore the signs $ ! $ and $ ? $ using the projection $ \projMess $.
Then $ \LangTree{\Automaton{p}} $ contains the words of $ \Automaton{p} $ that respect the possible input sequences induced by the parent $ q $.

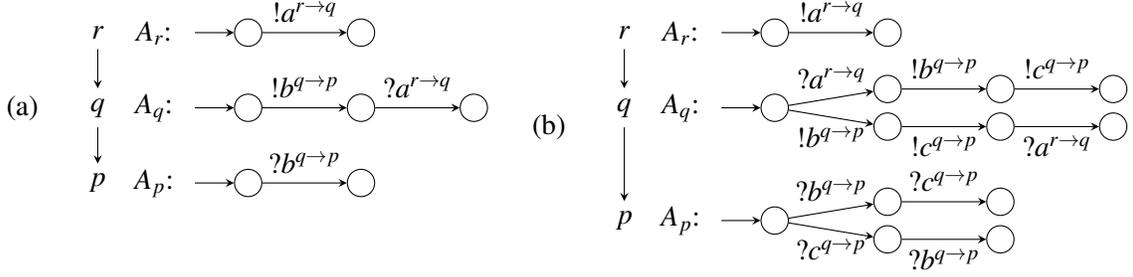
\begin{figure}[t]
	\centering
	\begin{tikzpicture}[auto, >=stealth, ->, every state/.style={minimum size = 0pt}, initial text={}]
		\node (a) at (-1, 1.5) {(a)};
		\node (r) at (0, 2.5) {$ \rootP $};
		\node (p) at (0, 1.5) {$ q $};
		\node (q) at (0, 0.5) {$ p $};
		\path (r) edge (p);
		\path (p) edge (q);
		\node (Ar) at (0.75, 2.5) {$ \Automaton{\rootP} $:};
		\node[state, initial] (s0) at (2, 2.5) {};
		\node[state] (s1) at (3.5, 2.5) {};
		\path (s0) edge node {$ \send{a}{\rootP}{q} $} (s1);
		\node (Aq) at (0.75, 1.5) {$ \Automaton{q} $:};
		\node[state, initial] (s0) at (2, 1.5) {};
		\node[state] (s1) at (3.5, 1.5) {};
		\node[state] (s2) at (5, 1.5) {};
		\path (s0) edge node {$ \send{b}{q}{p} $} (s1);
		\path (s1) edge node {$ \rec{a}{\rootP}{q} $} (s2);
		\node (Ap) at (0.75, 0.5) {$ \Automaton{p} $:};
		\node[state, initial] (s0) at (2, 0.5) {};
		\node[state] (s1) at (3.5, 0.5) {};
		\path (s0) edge node {$ \rec{b}{q}{p} $} (s1);
		\node (a) at (6, 1.25) {(b)};
		\node (r) at (7, 2.5) {$ \rootP $};
		\node (p) at (7, 1.5) {$ q $};
		\node (q) at (7, 0) {$ p $};
		\path (r) edge (p);
		\path (p) edge (q);
		\node (Ar) at (7.75, 2.5) {$ \Automaton{\rootP} $:};
		\node[state, initial] (s0) at (9, 2.5) {};
		\node[state] (s1) at (10.5, 2.5) {};
		\path (s0) edge node {$ \send{a}{\rootP}{q} $} (s1);
		\node (Aq) at (7.75, 1.5) {$ \Automaton{q} $:};
		\node[state, initial] (s0) at (9, 1.5) {};
		\node[state] (s1) at (10.5, 1.75) {};
		\node[state] (s2) at (12, 1.75) {};
		\node[state] (s3) at (13.5, 1.75) {};
		\node[state] (s4) at (10.5, 1.25) {};
		\node[state] (s5) at (12, 1.25) {};
		\node[state] (s6) at (13.5, 1.25) {};
		\path (s0) edge node[above] {$ \rec{a}{\rootP}{q} $} (s1);
		\path (s0) edge node[below] {$ \send{b}{q}{p} $} (s4);
		\path (s1) edge node {$ \send{b}{q}{p} $} (s2);
		\path (s2) edge node {$ \send{c}{q}{p} $} (s3);
		\path (s4) edge node[below] {$ \send{c}{q}{p} $} (s5);
		\path (s5) edge node[below] {$ \rec{a}{\rootP}{q} $} (s6);
		\node (Ap) at (7.75, 0) {$ \Automaton{p} $:};
		\node[state, initial] (s0) at (9, 0) {};
		\node[state] (s1) at (10.5, 0.25) {};
		\node[state] (s2) at (12, 0.25) {};
		\node[state] (s3) at (10.5, -0.25) {};
		\node[state] (s4) at (12, -0.25) {};
		\path (s0) edge node[above] {$ \rec{b}{q}{p} $} (s1);
		\path (s0) edge node[below] {$ \rec{c}{q}{p} $} (s3);
		\path (s1) edge node {$ \rec{c}{q}{p} $} (s2);
		\path (s3) edge node[below] {$ \rec{b}{q}{p} $} (s4);
	\end{tikzpicture}
	\caption{Examples for dependencies that prevent synchronisability}
	\label{fig:notSynchronisable}
\end{figure}

\begin{example}
	Figure~\ref{fig:notSynchronisable} depicts two examples of networks with their topology and the automata of each participant. 
	$ \LangTree{\Automaton{p}} $ only rules out paths that do not respect the sends of its parents.
	Hence, $ \LangTree{\Automaton{\rootP}} = \languageof{\Automaton{\rootP}} = \Set{\emptyword, \send{a}{\rootP}{q}} $, $ \LangTree{\Automaton{q}} = \languageof{\Automaton{q}} = \Set{\emptyword, \send{b}{q}{p}, \send{b}{q}{p} \rec{a}{\rootP}{q}} $, and $ \LangTree{\Automaton{p}} = \languageof{\Automaton{p}} = \Set{\emptyword, \rec{b}{q}{p}} $ in Figure~\ref{fig:notSynchronisable}.(a), but in Figure~\ref{fig:notSynchronisable}.b) $ \LangTree{\Automaton{p}} = \languageof{\Automaton{p}} \setminus \Set{ \rec{c}{q}{p}, \rec{c}{q}{p} \rec{b}{q}{p} } $.
\end{example}

In synchronous communication, sends and receptions are blocking, \ie they have to wait for matching communication partners.
In asynchronous communication with unbounded buffers, only inputs are blocking, whereas all outputs can be performed immediately.
Hence, for synchronisability the automata should provide matching inputs whenever their communication partners are ready to send.
We use causality to check for this condition.
For some automaton $ \Automaton{p} $, action $ a_2 $ \emph{causally depends} on action $ a_1 $, denoted as $ a_1 <_{p} a_2 $, if for all $ w \in \LangTree{\Automaton{p}} $ action $ a_2 $ does not occur or $ a_1 $ occurs before $ a_2 $.

First, we have to check  that in no automata we find a relation of the form $ !x <_p ?y $, because such a dependency always leads to non-synchronisability.
Intuitively, with $ !x <_p ?y $, the automaton $ p $ enforces the order $ !x $ before $ !y $ in the synchronous language, whereas in the asynchronous case with unbounded buffers, these sends may occur in any order.
Not having $ !x <_p ?y $ means that if $ !x $ can occur before $ ?y $ in $ A_p $, then another path in $ A_p $ allows to have $ ?y $ before $ !x $. 

\begin{example}
	Consider Figure~\ref{fig:notSynchronisable}.(a), $ \System{N}{\synch}{} $  has to perform $ \send{b}{q}{p} $ before $ \send{a}{\rootP}{q} $, because $ \rec{a}{\rootP}{q} $ is initially not available in $ \Automaton{q} $.
	In $ \System{N}{\mailbox}{} $, we have the execution $ \send{a}{\rootP}{q} \send{b}{q}{p} \rec{a}{\rootP}{q} \rec{b}{q}{p} $ and hence the trace $ \send{a}{\rootP}{q} \send{b}{q}{p} $.
	The problem is the dependency $ \send{b}{q}{p} <_q \rec{a}{\rootP}{q} $, that blocks the $ a $ in the synchronous but not the asynchronous system.
\end{example}

The other three kinds of dependencies are not necessarily problematic.
Causal dependencies of the form $ ?x <_p ?y $ are enforced by the parent of $ p $.
Dependencies of the form $ ?x <_p !y $ allow $ p $ to make its behaviour depending on the input of its parent.
Finally, dependencies $ !x <_p !y $ allow $ p $ to implement a certain strategy on sends.

Note that, $ \Automaton{p} $ may perform an action $ a $ several times.
Thus, we count the occurrences of $ a $ in a word such that $ a <_p a' $ is actually $ a_{\# n} <_p a'_{\# m} $, where $ a_{\# i} $ is the $ i $'th occurrence of $ a $.
This allows to express dependencies such as $ a_{\# 2} <_p a_{\# 3} $ (the third $ a $ depends on the second) or $ a_{\# 2} <_p a'_{\# 1} $ ($ a' $ depends on the second $ a $).
However, we keep the counters implicit, if they are not relevant, \ie if every action is unique.
For instance all actions in the examples in Figure~\ref{fig:notSynchronisable} are unique and thus we do not mention any counters.

Then, we have to check that missing inputs cannot block outputs of a parent.
The word  $ w' $ is a \emph{valid input shuffle of} $ w $, denoted as $ w' \shuffle_{?} w $, if $ w' $ is obtained from $ w $ by a (possibly empty) number of swappings that replace some $ !x?y $ within $ w $ by $ ?y!x $.

\begin{definition}[Shuffled language]
	Let $ p \in \Part $.
	We define its shuffled language as follows:
	$$ \LangShuffle{p} = \Set{w' \mid w \in \LangTree{\Automaton{p}} \wedge w' \shuffle_{?} w} $$ 
\end{definition}

Then we require $ \LangTree{\Automaton{p}} = \LangShuffle{p} $ to ensure synchronisability, and more precisely to avoid to have any dependence $ !x <_p ?y $ in a participant $ p $.
Note that $ \shuffle_? $ only allows to move inputs further to the front by swapping them with outputs.
Neither the order of outputs nor of inputs within the word is changed.

\begin{example}
	Consider Figure~\ref{fig:notSynchronisable}.(b).
	This example is not synchronisable, because $ \send{b}{q}{p} \send{a}{\rootP}{q} \send{c}{q}{p} \in \traceSet{\System{N}{\mailbox}{}} $ but $ \send{b}{q}{p} \send{a}{\rootP}{q} \send{c}{q}{p} \notin \traceSet{\System{N}{\synch}{}} $.
	Indeed the condition $ \LangTree{\Automaton{p}} = \LangShuffle{p} $ is violated:
	\begin{align*}
		\LangTree{\Automaton{q}} &= \Set{\emptyword, \rec{a}{\rootP}{q}, \send{b}{q}{p}, \rec{a}{\rootP}{q} \send{b}{q}{p}, \send{b}{q}{p} \send{c}{q}{p}, \rec{a}{\rootP}{q} \send{b}{q}{p} \send{c}{q}{p}, \send{b}{q}{p} \send{c}{q}{p} \rec{a}{\rootP}{q}}\\
		&\neq \LangShuffle{q}
		= \LangTree{\Automaton{q}} \cup \Set{ \send{b}{q}{p} \rec{a}{\rootP}{q}, \send{b}{q}{p} \rec{a}{\rootP}{q} \send{c}{q}{p} }
	\end{align*}
	Since $ \Automaton{q} $ does not allow for all possible valid input shufflings, after $ b $ the action $ a $ becomes blocked in $ \Automaton{q} $ in the synchronous but not the asynchronous system.
\end{example}

Finally, we have to check that for each send of a parent there is a matching input in the child, \ie $ \LangTreeOut{\Automaton{q}}\projPeers{\Set{p, q}}\projMess \subseteq \languageof{\Automaton{p}}\projMess $ whenever $ \senders{p} = \Set{q} $.
Unmatched sends appear as sends in asynchronous languages, but are not present in synchronous languages.
This is the trick that we have used in the Post Correspondence encoding to force the synchronous set of traces to be empty.


Note that $ \LangTreeOut{\Automaton{q}}\projPeers{\Set{p, q}}\projMess \subseteq \LangTree{\Automaton{p}}\projMess \wedge \LangTree{\Automaton{p}} = \LangShuffle{p} $ ensures all three conditions, \ie also ensures that there are no dependencies of the form $ !x <_p ?y $.

We prove first that words in $ \LangTree{\Automaton{q}} $ belong to executions of the \Mailbox system with unbounded buffers that do not require any interaction with a child of $ q $.

\begin{lemma}\label{lem:wordsoftreelanguage}
	Let $ \Network $ be a network such that $\configFinal = \configSet$, $ \topo{\Network} $ is a tree, $ q \in \Part $, and $ w \in \LangTree{\Automaton{q}} $.
	Then there is an execution $ w' \in \execSet{\System{\Network}{\mailbox}{}} $ such that $ w'\projPeer{q} = w $ and $ w'\projPeer{p} = \varepsilon $ for all $ p \in \Part $ with $ \senders{p} = \Set{q} $.
\end{lemma}

\begin{proof}
	We construct $ w' $ from the unique path from the root $ \rootP $ to $ q $.
	Let $ \rootP = q_1 $, $ q_2 $, \ldots, $ q_n = q $ be this path of length $ n $ such that $ \senders{q_i} = \Set{q_{i - 1}} $ for all $ 1 < i \leq n $.
	Remember that in a tree there is exactly one path from the root to every node.
	Hence, this path $ \rootP = q_1 $, $ q_2 $, \ldots, $ q_n = q $ and its length $ n $ are uniquely defined by $ q $.
	In the following, let $ w_n = w $.
	For $ n = 1 $, \ie for the case $ q = \rootP $ and $ w \in \LangTree{\Automaton{q}} $, $ w $ consists of outputs only.
	In this case we can choose $ w' = w $ such that $ w' \in \execSet{\System{\Network}{\mailbox}{}} $, $ w'\projPeer{q} = w $, and $ w'\projPeer{p} = \varepsilon $ for all $ p \in \Part $ with $ \senders{p} = \Set{q} $.
	It remains to show that these conditions are satisfied if $ n > 1 $, \ie $ q \neq \rootP $.
	Because of $ w \in \LangTree{\Automaton{q}} $, there is some $ w_{n - 1} \in \LangTree{\Automaton{q_{n - 1}}} $ such that $ \left( w_{n - 1}\projPeer{q} \right)\projMess = \left( w\projPeer{q_{n - 1}} \right)\projMess $, \ie $ w_{n - 1} $ provides the outputs for all inputs in $ w_n $ in the required order.
	By repeating this argument moving from $ q $ towards the root along the path, there is some $ w_{i - 1} \in \LangTree{\Automaton{q_{i - 1}}} $ such that $ \left( w_{i - 1}\projPeer{q_i} \right)\projMess = \left( w_i\projPeer{q_{i - 1}} \right)\projMess $ for all $ 1 < i \leq n $.
	Then $ w' = w_1 w_2 \ldots w_n $.
	Since $ q_1 = \rootP $ is a root, $ w_1 $ contains only outputs, \ie $ w_1 \in \execSet{\System{\Network}{\mailbox}{}} $.
	For all inputs in $ w_2 $, $ w_1 $ provides the matching outputs in the correct order, \ie $ w_1w_2 \in \execSet{\System{\Network}{\mailbox}{}} $.
	By repeating this argument moving from $ \rootP $ towards $ q $ along the path, then $ w' = w_1 \ldots w_n \in \execSet{\System{\Network}{\mailbox}{}} $.
	By construction, $ w'\projPeer{q} = w $ and $ w'\projPeer{p} = \varepsilon $ for all $ p \in \Part $ with $ \senders{p} = \Set{q} $.
\end{proof}

Finally, the next theorem states that synchronisability can be checked by verifying that for all neighbouring peers $p$ and its parent $q$, all sequences of sends from $q$ can be received by $p$ at any moment, \ie without blocking sends from $p$. 

\begin{theorem}
	Let $ \Network $ be a network such that $\configFinal = \configSet$ and $ \topo{\Network} $ is a tree.
	Then $ \traceSet{\System{\Network}{\mailbox}{}} = \traceSet{\System{\Network}{\synch}{}} $ iff, for all $ p, q \in \Part $ with $ \senders{p} = \Set{q} $, we have $ \left( \LangTreeOut{\Automaton{q}}\projPeers{\Set{p, q}} \right)\projMess \subseteq \LangTree{\Automaton{p}}\projMess $ and $ \LangTree{\Automaton{p}} = \LangShuffle{p} $.
	\label{thm:synchronisabilityTree}
\end{theorem}

\begin{proof}
	\begin{compactdesc}
		\item[$ \Rightarrow $] Assume $ \traceSet{\System{\Network}{\mailbox}{}} = \traceSet{\System{\Network}{\synch}{}} $.
			We have to show that 1. $ \left( \LangTreeOut{\Automaton{q}}\projPeers{\Set{p, q}} \right)\projMess \subseteq \LangTree{\Automaton{p}}\projMess $ and 2. $ \LangTree{\Automaton{p}} = \LangShuffle{p} $ for all $ p, q \in \Part $ with $ \senders{p} = \Set{q} $.
			\begin{compactenum}
				\item Assume $ \left( \LangTreeOut{\Automaton{q}}\projPeers{\Set{p, q}} \right)\projMess \not\subseteq \LangTree{\Automaton{p}}\projMess $ for some $ p, q \in \Part $ with $ \senders{p} = \Set{q} $.
					Then there is some sequence of outputs $ v \in \LangTreeOut{\Automaton{q}}\projPeers{\Set{p, q}} $ for that there is no matching sequence of inputs in $ \LangTree{\Automaton{p}} $, \ie $ v\projMess \notin \LangTree{\Automaton{p}}\projMess $.
					By Lemma~\ref{lem:wordsoftreelanguage}, then there is an execution $ v' \in \execSet{\System{\Network}{\mailbox}{}} $ such that $ \left( v'\projOut \right)\projPeer{q} = v $ and $ v'\projPeer{p} = \varepsilon $.
					Hence, $ v'\projOut \in \traceSet{\System{\Network}{\mailbox}{ }} = \traceSet{\System{\Network}{\synch}{}} = \execSet{\System{\Network}{\synch}{}} $.
					Because of $ v'\projOut \in \execSet{\System{\Network}{\synch}{}} $ and $ v'\projPeer{p} = \varepsilon $, $ \Automaton{p} $ has to be able to receive the sequence of outputs of $ v $ without performing any outputs itself, \ie $ v\projMess \in \left( \languageof{\Automaton{p}}\projIn \right)\projMess $ and $ v\projMess \in \LangTreeIn{\Automaton{p}}\projMess $.
					But then $ v\projMess \in \LangTree{\Automaton{p}}\projMess $.
					This is a contradiction.
					We conclude that $ \left( \LangTreeOut{\Automaton{q}}\projPeers{\Set{p, q}} \right)\projMess \subseteq \LangTree{\Automaton{p}}\projMess $ for all $ p, q \in \Part $ with $ \senders{p} = \Set{q} $.
				\item By definition, $ \LangTree{\Automaton{p}} \subseteq \LangShuffle{p} $.
					Assume $ v \in \LangShuffle{p} $.
					We have to show that $ v \in \LangTree{\Automaton{p}} $.
					Since $ v \in \LangShuffle{p} $, then there is some $ v' $ such that $ v' \in \LangTree{\Automaton{p}} $ and $ v' \shuffle_{?} v $.
					We consider the shortest two such words $ v $ and $ v' $, \ie $ v = w \rec{a}{q}{p} !x_1 \ldots !x_n $ and $ v' = w !x_1 \ldots !x_n \rec{a}{q}{p} $ with $ n > 0 $, where $ \senders{p} = \Set{q} $ and $ x_1, \ldots x_n $ are ouputs from $ p $ to its children.
					Then also $ w !x_1 \ldots !x_n \rec{a}{q}{p} \in \LangTree{\Automaton{p}} $.
					By Lemma~\ref{lem:wordsoftreelanguage}, then there is an execution $ w' \in \execSet{\System{\Network}{\mailbox}{}} $ such that $ w'\projPeer{p} = w !x_1 \ldots !x_n \rec{a}{q}{p} $.
					By the construction of $ w' $ in the proof of Lemma~\ref{lem:wordsoftreelanguage}, $ w' $ contains the output $ \send{a}{q}{p} $ before the outputs $ !x_1 \ldots !x_n $.
					Then $ \send{a}{q}{p} $ occurs before $ !x_1 \ldots !x_n $ in $ w'\projOut \in \traceSet{\System{\Network}{\mailbox}{}} $.
					Since $ \traceSet{\System{\Network}{\mailbox}{}} = \traceSet{\System{\Network}{\synch}{}} = \execSet{\System{\Network}{\synch}{}} $, $ w'\projOut \in \execSet{\System{\Network}{\synch}{}} $.
					Then $ \Automaton{p} $ has to receive $ \rec{a}{q}{p} $ before sending $ !x_1 \ldots !x_n $, \ie $ v = w \rec{a}{q}{p} !x_1 \ldots !x_n \in \LangTree{\Automaton{p}} $.
					We conclude that $ \LangTree{\Automaton{p}} = \LangShuffle{p} $.
			\end{compactenum}
		\item[$ \Leftarrow $] Assume that $ \left( \LangTreeOut{\Automaton{q}}\projPeers{\Set{p, q}} \right)\projMess \subseteq \LangTree{\Automaton{p}}\projMess $ and $ \LangTree{\Automaton{p}} = \LangShuffle{p} $ for all $ p, q \in \Part $ with $ \senders{p} = \Set{q} $.
			We have to show that $ \traceSet{\System{\Network}{\mailbox}{ }} = \traceSet{\System{\Network}{\synch}{}} $.
			\begin{compactdesc}
				\item[$ w \in \traceSet{\System{\Network}{\mailbox}{}} $:] Let $ w' $ be the word obtained from $ w $ by adding the matching receive action directly after every send action.
					We show that $ w' \in \execSet{\System{\Network}{\mailbox}{ }} $, by an induction on the length of $ w $.
					\begin{compactdesc}
						\item[Base Case:] If $ w = \send{a}{q}{p} $, then $ w' = \send{a}{q}{p} \rec{a}{q}{p} $.
							Since $ w \in \traceSet{\System{\Network}{\mailbox}{ }} $, $ \Automaton{q} $ is able to send $ \send{a}{q}{p} $ in its initial state within the system $ \System{\Network}{\mailbox}{ } $.
							Then $ \send{a}{q}{p} \in \LangTreeOut{\Automaton{q}} $.
							Because of $ \left( \LangTreeOut{\Automaton{q}}\projPeers{\Set{p, q}} \right)\projMess \subseteq \LangTree{\Automaton{p}}\projMess $, then $ \rec{a}{q}{p} \in \LangTree{\Automaton{p}} $, \ie $ \Automaton{q} $ can receive $ \rec{a}{q}{p} $ in its initial state.
							Then $ w' \in \execSet{\System{\Network}{\mailbox}{ }} $.
						\item[Inductive Step:] If $ w = v \send{a}{q}{p} $ with $ v \send{a}{q}{p} \in \traceSet{\System{\Network}{\mailbox}{ }} $, then $ w' = v' \send{a}{q}{p} \rec{a}{q}{p} $.
							By  induction, $ v' \in \execSet{\System{\Network}{\mailbox}{ }} $.
							Since $ w = v \send{a}{q}{p} \in \traceSet{\System{\Network}{\mailbox}{ }} $, $ \Automaton{q} $ is able to perform $ \send{a}{q}{p} $ in some state after performing all the outputs in $ v $.
							Since $ v' \projOut = v $, then $ \Automaton{q} $ is able to perform $ \send{a}{q}{p} $ in some state after performing all the outputs in $ v' $.
							Also inputs of $ \Automaton{q} $ cannot prevent $ \Automaton{q} $ from sending $ \send{a}{q}{p} $ after $ v' $, because:
							\begin{itemize}
								\item For all such inputs $ ?y $ there has to be $ !y $ occurring before the input.
								\item Then $ !y $ is among the outputs of $ v $, because $ \Automaton{q} $ is able to perform $ \send{a}{q}{p} $ in some state after performing all the outputs in $ v $ of $ q $.
								\item Then also $ ?y $ is already contained in $ v' $, by the construction of $ v' $.
							\end{itemize}
							This entails that  $ \Automaton{q} $ can send $ \send{a}{q}{p} $ after execution $ v' $ such that $ v' \send{a}{q}{p} \in \execSet{\System{\Network}{\mailbox}{ }} $.
							Hence we have $ \left( \left( v' \send{a}{q}{p} \right) \projPeers{\Set{p, q}} \right) \projOut = \left( w \projPeers{\Set{p, q}} \right) \projOut \in \LangTreeOut{\Automaton{q}} \projPeers{\Set{p, q}} $ and after $ v' \send{a}{q}{p} $ all buffers are empty except for the buffer of $ \Automaton{p} $ that contains only $ \msg{a}{q}{p} $.
							By noticing that $ \left( \LangTreeOut{\Automaton{q}}\projPeers{\Set{p, q}} \right)\projMess \subseteq \LangTree{\Automaton{p}}\projMess $, then $ \left( \left( v' \rec{a}{q}{p} \right) \projPeers{\Set{p, q}} \right) \projIn \in \LangTree{\Automaton{p}} $, \ie $ \Automaton{p} $ is able to receive $ \rec{a}{q}{p} $ in some state after receiving all the inputs in $ v' $.
							By $ \LangTree{\Automaton{p}} = \LangShuffle{p} $ and since $ \msg{a}{q}{p} $ is in its buffer, then $ \Automaton{p} $ can receive $ \rec{a}{q}{p} $ after execution $ v' \send{a}{q}{p} $ such that $ w' = v' \send{a}{q}{p} \rec{a}{q}{p} \in \execSet{\System{N}{\mailbox}{ }} $.
					\end{compactdesc}
					Hence $ w' \in \execSet{\System{\Network}{\mailbox}{ }} $.
					This entails that the \Synchronous system can simulate the run of $ w' $ in $ \System{\Network}{\mailbox}{ } $ by combining a send action with its direct following matching receive action into a synchronous communication.
					Since $ w'\projOut = w $, then $ w \in \execSet{\System{\Network}{\synch}{}} = \traceSet{\System{\Network}{\synch}{}} $.
				\item[$ w \in \traceSet{\System{\Network}{\synch}{}} $:] For every output in $ w $, $ \System{N}{\synch}{} $ was able to send the respective message and directly receive it.
					Let $ w' $ be the word obtained from $ w $ by adding the matching receive action directly after every send action.
					Then $ \System{\Network}{\mailbox}{ } $ can simulate the run of $ w $ in $ \System{\Network}{\synch}{} $ by sending every message first to the \Mailbox of the receiver and then receiving this message.
					Then \mbox{$ w' \in \execSet{\System{\Network}{\mailbox}{ }} $} and, thus, $ w = w'\projOut \in \traceSet{\System{\Network}{\mailbox}{ }} $. \qedhere
			\end{compactdesc}
	\end{compactdesc}
\end{proof}

Since there is no difference between \Mailbox and \Ptp communication with a tree topology, we have $ \traceSet{\System{\Network}{\mailbox}{}} = \traceSet{\System{\Network}{\pp}{}} $.
Accordingly, Theorem~\ref{thm:synchronisabilityTree} provides a decision procedure for \Mailbox and \Ptp systems.
In both cases it suffice to algoritmically check, whether $ \left( \LangTreeOut{\Automaton{q}}\projPeers{\Set{p, q}} \right)\projMess \subseteq \LangTree{\Automaton{p}}\projMess $ and $ \LangTree{\Automaton{p}} = \LangShuffle{p} $ for all $ p, q \in \Part $ with $ \senders{p} = \Set{q} $.
This can be done by computing the influenced languages starting from the root moving down in the tree.

\begin{corollary}
	Let $ \Network $ be a network such that $\configFinal = \configSet$ and $ \topo{\Network} $ is a tree.
	Then the Synchronisability Problem is decidable for \Ptp and \Mailbox communication.
\end{corollary}

%
%
%
%
%

\section{Discussion}
\label{sec:concl}

In this paper, we start answering a problem that have remained open since \cite{finkel_synchronizability_2017}. We first have shown that the Generalised Synchronisability Problem is undecidable for \Mailbox systems.  The undecidability result cannot be easily adpated to communicating automata without final states as the role of automaton \Automaton{L} (\ie the comparator) is made more complex by the fact that letters that end up in its buffer are mixed between those coming from \Automaton{W} and \Automaton{W'}. This would require an additional synchronisation between \Automaton{W} and  \Automaton{W'} which would mess the exchanges between those automata and \Automaton{I}.

Hence, in an attempt to get closer to a proof of decidability for Synchronisability Problem, we considered tree topologies. 
We have presented an algorithm to decide synchronisability for systems that feature a tree topology.
The key ingredient in the above algorithm for trees is, that we can compute the languages $ \LangTree{\Automaton{p}} $ and thus the possible behaviour for every node, by starting from the root and following the unique path from the root to the respective node.
Then we only check properties on a single node ($ \LangTree{\Automaton{p}} = \LangShuffle{p} $) or between two neighbouring nodes ($ \LangTreeOut{\Automaton{q}}\projPeers{\Set{p, q}}\projMess \subseteq \LangTree{\Automaton{p}}\projMess $).

We conjecture that this result can be extended to reversed trees and multitrees \cite{furnas_multitrees_1994}. By observing that 
a forest is synchronisable if and only if each of its trees is synchronisable. 
Moreover reversed trees and multitrees have the same property of featuring  unique paths between any two nodes.
Hence, we conjecture that the above technique can be extended to reversed trees and multitrees.
Moreover the absence of coordination between brothers should also entail that the result is true for \Ptp.
These developments are left for future work.
Moreover, we are working on a mechanisation of our proofs in Isabelle, which is quite challenging as there are few existing mechanisation approaches for communicating automata.




\bibliography{biblio}

\begin{thebibliography}{10}
\providecommand{\bibitemdeclare}[2]{}
\providecommand{\surnamestart}{}
\providecommand{\surnameend}{}
\providecommand{\urlprefix}{Available at }
\providecommand{\url}[1]{\texttt{#1}}
\providecommand{\href}[2]{\texttt{#2}}
\providecommand{\urlalt}[2]{\href{#1}{#2}}
\providecommand{\doi}[1]{doi:\urlalt{https://doi.org/#1}{#1}}
\providecommand{\eprint}[1]{arXiv:\urlalt{https://arxiv.org/abs/#1}{#1}}
\providecommand{\bibinfo}[2]{#2}

\bibitemdeclare{article}{basu_deciding_2016}
\bibitem{basu_deciding_2016}
\bibinfo{author}{Samik \surnamestart Basu\surnameend} \&
  \bibinfo{author}{Tevfik \surnamestart Bultan\surnameend}
  (\bibinfo{year}{2016}): \emph{\bibinfo{title}{On deciding synchronizability
  for asynchronously communicating systems}}.
\newblock {\slshape \bibinfo{journal}{Theoretical Computer Science}}
  \bibinfo{volume}{656}, pp. \bibinfo{pages}{60--75},
  \doi{10.1016/j.tcs.2016.09.023}.
\newblock
  \urlprefix\url{http://www.sciencedirect.com/science/article/pii/S0304397516305102}.

\bibitemdeclare{inproceedings}{bollig_unifying_2021}
\bibitem{bollig_unifying_2021}
\bibinfo{author}{Benedikt \surnamestart Bollig\surnameend},
  \bibinfo{author}{Cinzia \surnamestart {Di Giusto}\surnameend},
  \bibinfo{author}{Alain \surnamestart Finkel\surnameend},
  \bibinfo{author}{Laetitia \surnamestart Laversa\surnameend},
  \bibinfo{author}{{\'{E}}tienne \surnamestart Lozes\surnameend} \&
  \bibinfo{author}{Amrita \surnamestart Suresh\surnameend}
  (\bibinfo{year}{2021}): \emph{\bibinfo{title}{A Unifying Framework for
  Deciding Synchronizability}}.
\newblock In \bibinfo{editor}{Serge \surnamestart Haddad\surnameend} \&
  \bibinfo{editor}{Daniele \surnamestart Varacca\surnameend}, editors:
  {\slshape \bibinfo{booktitle}{32nd International Conference on Concurrency
  Theory, {CONCUR} 2021, August 24-27, 2021, Virtual Conference}}, {\slshape
  \bibinfo{series}{LIPIcs}} \bibinfo{volume}{203}, \bibinfo{publisher}{Schloss
  Dagstuhl - Leibniz-Zentrum f{\"{u}}r Informatik}, pp.
  \bibinfo{pages}{14:1--14:18}, \doi{10.4230/LIPICS.CONCUR.2021.14}.

\bibitemdeclare{inproceedings}{bouajjani_completeness_2018}
\bibitem{bouajjani_completeness_2018}
\bibinfo{author}{Ahmed \surnamestart Bouajjani\surnameend},
  \bibinfo{author}{Constantin \surnamestart Enea\surnameend},
  \bibinfo{author}{Kailiang \surnamestart Ji\surnameend} \&
  \bibinfo{author}{Shaz \surnamestart Qadeer\surnameend}
  (\bibinfo{year}{2018}): \emph{\bibinfo{title}{On the Completeness of
  Verifying Message Passing Programs Under Bounded Asynchrony}}.
\newblock In \bibinfo{editor}{Hana \surnamestart Chockler\surnameend} \&
  \bibinfo{editor}{Georg \surnamestart Weissenbacher\surnameend}, editors:
  {\slshape \bibinfo{booktitle}{Computer Aided Verification - 30th
  International Conference, {CAV} 2018, Held as Part of the Federated Logic
  Conference, FloC 2018, Oxford, UK, July 14-17, 2018, Proceedings, Part
  {II}}}, {\slshape \bibinfo{series}{Lecture Notes in Computer Science}}
  \bibinfo{volume}{10982}, \bibinfo{publisher}{Springer}, pp.
  \bibinfo{pages}{372--391}, \doi{10.1007/978-3-319-96142-2\_23}.

\bibitemdeclare{article}{brand_communicating_1983}
\bibitem{brand_communicating_1983}
\bibinfo{author}{Daniel \surnamestart Brand\surnameend} \&
  \bibinfo{author}{Pitro \surnamestart Zafiropulo\surnameend}
  (\bibinfo{year}{1983}): \emph{\bibinfo{title}{On {Communicating}
  {Finite}-{State} {Machines}}}.
\newblock {\slshape \bibinfo{journal}{Journal of the ACM}}
  \bibinfo{volume}{30}(\bibinfo{number}{2}), pp. \bibinfo{pages}{323--342},
  \doi{10.1145/322374.322380}.

\bibitemdeclare{article}{charron_synchronous_1996}
\bibitem{charron_synchronous_1996}
\bibinfo{author}{Bernadette \surnamestart Charron{-}Bost\surnameend},
  \bibinfo{author}{Friedemann \surnamestart Mattern\surnameend} \&
  \bibinfo{author}{Gerard \surnamestart Tel\surnameend} (\bibinfo{year}{1996}):
  \emph{\bibinfo{title}{Synchronous, Asynchronous, and Causally Ordered
  Communication}}.
\newblock {\slshape \bibinfo{journal}{Distributed Comput.}}
  \bibinfo{volume}{9}(\bibinfo{number}{4}), pp. \bibinfo{pages}{173--191},
  \doi{10.1007/S004460050018}.

\bibitemdeclare{article}{chevrou_diversity_2016}
\bibitem{chevrou_diversity_2016}
\bibinfo{author}{Florent \surnamestart Chevrou\surnameend},
  \bibinfo{author}{Aur{\'{e}}lie \surnamestart Hurault\surnameend} \&
  \bibinfo{author}{Philippe \surnamestart Qu{\'{e}}innec\surnameend}
  (\bibinfo{year}{2016}): \emph{\bibinfo{title}{On the diversity of
  asynchronous communication}}.
\newblock {\slshape \bibinfo{journal}{Formal Aspects Comput.}}
  \bibinfo{volume}{28}(\bibinfo{number}{5}), pp. \bibinfo{pages}{847--879},
  \doi{10.1007/S00165-016-0379-X}.

\bibitemdeclare{article}{digiusto_partial_2023}
\bibitem{digiusto_partial_2023}
\bibinfo{author}{Cinzia \surnamestart {Di Giusto}\surnameend},
  \bibinfo{author}{Davide \surnamestart Ferr{\'{e}}\surnameend},
  \bibinfo{author}{Laetitia \surnamestart Laversa\surnameend} \&
  \bibinfo{author}{{\'{E}}tienne \surnamestart Lozes\surnameend}
  (\bibinfo{year}{2023}): \emph{\bibinfo{title}{A Partial Order View of
  Message-Passing Communication Models}}.
\newblock {\slshape \bibinfo{journal}{Proc. {ACM} Program. Lang.}}
  \bibinfo{volume}{7}(\bibinfo{number}{{POPL}}), pp.
  \bibinfo{pages}{1601--1627}, \doi{10.1145/3571248}.

\bibitemdeclare{inproceedings}{digiusto_k-synchronizability_2020}
\bibitem{digiusto_k-synchronizability_2020}
\bibinfo{author}{Cinzia \surnamestart {Di Giusto}\surnameend},
  \bibinfo{author}{Laetitia \surnamestart Laversa\surnameend} \&
  \bibinfo{author}{{\'{E}}tienne \surnamestart Lozes\surnameend}
  (\bibinfo{year}{2020}): \emph{\bibinfo{title}{On the k-synchronizability of
  Systems}}.
\newblock In \bibinfo{editor}{Jean \surnamestart Goubault{-}Larrecq\surnameend}
  \& \bibinfo{editor}{Barbara \surnamestart K{\"{o}}nig\surnameend}, editors:
  {\slshape \bibinfo{booktitle}{Foundations of Software Science and Computation
  Structures - 23rd International Conference, {FOSSACS} 2020, Held as Part of
  the European Joint Conferences on Theory and Practice of Software, {ETAPS}
  2020, Dublin, Ireland, April 25-30, 2020, Proceedings}}, {\slshape
  \bibinfo{series}{Lecture Notes in Computer Science}} \bibinfo{volume}{12077},
  \bibinfo{publisher}{Springer}, pp. \bibinfo{pages}{157--176},
  \doi{10.1007/978-3-030-45231-5\_9}.

\bibitemdeclare{inproceedings}{finkel_synchronizability_2017}
\bibitem{finkel_synchronizability_2017}
\bibinfo{author}{Alain \surnamestart Finkel\surnameend} \&
  \bibinfo{author}{Etienne \surnamestart Lozes\surnameend}
  (\bibinfo{year}{2017}): \emph{\bibinfo{title}{Synchronizability of
  {Communicating} {Finite} {State} {Machines} is not {Decidable}}}.
\newblock In \bibinfo{editor}{Ioannis \surnamestart
  Chatzigiannakis\surnameend}, \bibinfo{editor}{Piotr \surnamestart
  Indyk\surnameend}, \bibinfo{editor}{Fabian \surnamestart Kuhn\surnameend} \&
  \bibinfo{editor}{Anca \surnamestart Muscholl\surnameend}, editors: {\slshape
  \bibinfo{booktitle}{44th {International} {Colloquium} on {Automata},
  {Languages}, and {Programming} ({ICALP} 2017)}}, {\slshape
  \bibinfo{series}{Leibniz {International} {Proceedings} in {Informatics}
  ({LIPIcs})}}~\bibinfo{volume}{80}, \bibinfo{publisher}{Schloss
  Dagstuhl–Leibniz-Zentrum fuer Informatik}, \bibinfo{address}{Dagstuhl,
  Germany}, pp. \bibinfo{pages}{122:1--122:14},
  \doi{10.4230/LIPIcs.ICALP.2017.122}.
\newblock \urlprefix\url{http://drops.dagstuhl.de/opus/volltexte/2017/7402}.
\newblock \bibinfo{note}{ISSN: 1868-8969}.

\bibitemdeclare{inproceedings}{furnas_multitrees_1994}
\bibitem{furnas_multitrees_1994}
\bibinfo{author}{George~W. \surnamestart Furnas\surnameend} \&
  \bibinfo{author}{Jeff \surnamestart Zacks\surnameend} (\bibinfo{year}{1994}):
  \emph{\bibinfo{title}{Multitrees: enriching and reusing hierarchical
  structure}}.
\newblock In: {\slshape \bibinfo{booktitle}{Proceedings of the SIGCHI
  Conference on Human Factors in Computing Systems}}, \bibinfo{series}{CHI
  '94}, \bibinfo{publisher}{Association for Computing Machinery},
  \bibinfo{address}{New York, NY, USA}, p. \bibinfo{pages}{330–336},
  \doi{10.1145/191666.191778}.

\bibitemdeclare{article}{genest_kleene_2006}
\bibitem{genest_kleene_2006}
\bibinfo{author}{Blaise \surnamestart Genest\surnameend},
  \bibinfo{author}{Dietrich \surnamestart Kuske\surnameend} \&
  \bibinfo{author}{Anca \surnamestart Muscholl\surnameend}
  (\bibinfo{year}{2006}): \emph{\bibinfo{title}{A {Kleene} theorem and model
  checking algorithms for existentially bounded communicating automata}}.
\newblock {\slshape \bibinfo{journal}{Information and Computation}}
  \bibinfo{volume}{204}(\bibinfo{number}{6}), pp. \bibinfo{pages}{920--956},
  \doi{10.1016/j.ic.2006.01.005}.
\newblock
  \urlprefix\url{http://www.sciencedirect.com/science/article/pii/S0890540106000290}.

\bibitemdeclare{article}{giusto_guessing_2021}
\bibitem{giusto_guessing_2021}
\bibinfo{author}{Cinzia~Di \surnamestart Giusto\surnameend},
  \bibinfo{author}{Laetitia \surnamestart Laversa\surnameend} \&
  \bibinfo{author}{{\'{E}}tienne \surnamestart Lozes\surnameend}
  (\bibinfo{year}{2023}): \emph{\bibinfo{title}{Guessing the Buffer Bound for
  k-Synchronizability}}.
\newblock {\slshape \bibinfo{journal}{Int. J. Found. Comput. Sci.}}
  \bibinfo{volume}{34}(\bibinfo{number}{8}), pp. \bibinfo{pages}{1051--1076},
  \doi{10.1142/S0129054122430018}.

\bibitemdeclare{incollection}{kuske_communicating_2010}
\bibitem{kuske_communicating_2010}
\bibinfo{author}{Dietrich \surnamestart Kuske\surnameend} \&
  \bibinfo{author}{Anca \surnamestart Muscholl\surnameend}
  (\bibinfo{year}{2021}): \emph{\bibinfo{title}{Communicating automata}}.
\newblock In \bibinfo{editor}{Jean{-}{\'{E}}ric \surnamestart Pin\surnameend},
  editor: {\slshape \bibinfo{booktitle}{Handbook of Automata Theory}},
  \bibinfo{publisher}{European Mathematical Society Publishing House,
  Z{\"{u}}rich, Switzerland}, pp. \bibinfo{pages}{1147--1188},
  \doi{10.4171/AUTOMATA-2/9}.

\bibitemdeclare{article}{lamport_time_1978}
\bibitem{lamport_time_1978}
\bibinfo{author}{Leslie \surnamestart Lamport\surnameend}
  (\bibinfo{year}{1978}): \emph{\bibinfo{title}{Time, Clocks, and the Ordering
  of Events in a Distributed System}}.
\newblock {\slshape \bibinfo{journal}{Commun. {ACM}}}
  \bibinfo{volume}{21}(\bibinfo{number}{7}), pp. \bibinfo{pages}{558--565},
  \doi{10.1145/359545.359563}.

\bibitemdeclare{article}{post_variant_1946}
\bibitem{post_variant_1946}
\bibinfo{author}{Emil~L. \surnamestart Post\surnameend} (\bibinfo{year}{1946}):
  \emph{\bibinfo{title}{A variant of a recursively unsolvable problem}}.
\newblock {\slshape \bibinfo{journal}{Bulletin of the American Mathematical
  Society}} \bibinfo{volume}{52}(\bibinfo{number}{4}), pp.
  \bibinfo{pages}{264--268}, \doi{10.1090/S0002-9904-1946-08555-9}.
\newblock
  \urlprefix\url{https://www.ams.org/bull/1946-52-04/S0002-9904-1946-08555-9/}.

\end{thebibliography}



\end{document}